\documentclass[11pt,a4paper]{article}

\usepackage[T1]{fontenc}
\usepackage[utf8]{inputenc}
\usepackage{lmodern}
\usepackage{microtype}
\usepackage{geometry}
\usepackage{amsmath,amssymb,amsthm}
\usepackage{bm}
\usepackage{tikz}
\usepackage{pgfplots}
\pgfplotsset{compat=1.18}
\usepackage{graphicx}
\usepackage{booktabs}
\usepackage{siunitx}

\usepackage[space,nosort]{cite}

\usepackage[hidelinks]{hyperref}

\theoremstyle{plain}
\newtheorem{theorem}{Theorem}
\newtheorem{lemma}{Lemma}
\newtheorem{proposition}{Proposition}

\theoremstyle{definition}
\newtheorem{definition}{Definition}
\newtheorem{remark}{Remark}

\newcommand{\dd}{\mathrm{d}}

\newcommand{\ii}{\mathrm{i}}
\newcommand{\kb}{k_{\mathrm{B}}}
\newcommand{\rs}{r_{\mathrm{s}}}
\newcommand{\SBH}{S_{\mathrm{BH}}}
\newcommand{\FHelm}{F}

\title{\textbf{Energy Layers and Quasi-Superradiant Heat Engines of Schwarzschild Black Holes}}
\author{Wen-Xiang Chen \\
School of Physics and Materials Science, Guangzhou University\\wxchen4277@qq.com}
\date{\today}

\begin{document}
\maketitle

\begin{abstract}
We examine Schwarzschild black holes within the framework of gravitational thermodynamics, introducing an ``energy layer'' picture for black-hole mass-energy and exploring a possible energy-extraction mechanism termed ``quasi-superradiance.'' Building on the standard relations for Hawking temperature and Bekenstein--Hawking entropy, we formalize energy layers via quasi-local radial energy accounting (e.g.\ integrating an effective local energy density over spherical shells) and connect this bookkeeping to the free energy $\FHelm=M-\TH \SBH$. We then extend the entropy correction ansatz with explicit series inversion and derive higher-order expansions for $\TH(M)$ and $\FHelm(M)$, including logarithmic and inverse-mass terms. To enhance mathematical transparency, we add intermediate derivations, lemma/theorem statements, and appendices. The quasi-superradiant mechanism is framed as a Carnot-like thought experiment powered by the Tolman temperature gradient between the near-horizon region and infinity; we show that the generalized second law enforces the Carnot bound and yields integrated maximum-work inequalities. Throughout, we stress that the proposal is heuristic and intended as a consistency-checked framework for discussion rather than a claim of definitive new physics.
\end{abstract}


\section{Introduction}

Black hole thermodynamics establishes a deep connection between general relativity, quantum physics, and classical thermodynamic laws. Following the pioneering works of Bekenstein and Hawking in the 1970s \cite{ref1,ref2}, a Schwarzschild black hole of mass $M$ is assigned a temperature $\TH$ and an entropy $\SBH$ given (in units where $G$, $c$, $\hbar$, $\kb \equiv 1$ unless otherwise stated) by
\begin{equation}
\TH = \frac{1}{8\pi M}, \qquad \SBH = 4\pi M^{2},
\tag{1}
\end{equation}
which satisfy the first law $\dd M = \TH\,\dd \SBH$ and imply the second law in the form of the generalized entropy (black hole plus outside) never decreasing \cite{ref1}. The thermodynamic analogies are reinforced by the laws of black hole mechanics \cite{ref3}: the surface gravity $\kappa$ plays the role of temperature ($\TH=\kappa/2\pi$) and the horizon area $A$ plays the role of entropy ($\SBH=A/4$).

One intriguing aspect of Schwarzschild black holes is that, despite having no classical ``ergoregion'' (unlike a Kerr black hole which exhibits superradiant wave amplification due to its rotation \cite{ref4}), the combination of gravity and thermodynamics allows for energy exchange between the black hole and its surroundings. In particular, a Schwarzschild black hole radiates quantum mechanically (Hawking radiation) as if it were a black body at temperature $\TH$ \cite{ref2}. For a classic set of notes emphasizing the evaporation viewpoint, see \cite{ref6}. A striking feature is the gravitational redshift of this radiation: an observer at radius $r$ sees a higher local temperature $T_{\mathrm{local}}(r)=\TH/\sqrt{1-2M/r}$ due to Tolman's law of thermal equilibrium in a static gravitational field \cite{ref3}. As $r\to 2M^{+}$, $T_{\mathrm{local}}\to \infty$, meaning the region just outside the horizon (sometimes termed the black hole's ``atmosphere'') is extremely hot compared to the cooler far-away region. This enormous temperature gradient invites the question of whether it can be exploited to extract useful work, much as a classical heat engine operates between a hot and cold reservoir.

In this paper, we explore two connected ideas in this context: (i) a stratified energy layer structure of the Schwarzschild black hole, and (ii) a quasi-superradiant thermodynamic mechanism for energy extraction using the aforementioned temperature gradient. The term ``energy layer'' is used here as a descriptive model wherein different radial zones around the black hole contribute distinct amounts to the total effective mass-energy. To give the concept physical substance, we propose an operational definition in terms of the radial accumulation of field energy, and we relate it to a controlled free-energy expansion.

The second concept, quasi-superradiance, refers to an energy-extraction process that we propose as an analogue to superradiant scattering. In a rotating (Kerr) black hole, superradiance allows waves of certain frequencies to extract rotational energy from the black hole, a phenomenon requiring an ergoregion \cite{ref4}. A non-rotating (Schwarzschild) black hole does not allow classical superradiance because it has no ergoregion. However, thermodynamically, one can imagine setting up a Carnot engine between the near-horizon region (at $T_{\mathrm{local}}\gg \TH$) and infinity (at $T_{\infty}\approx \TH$ or some colder sink) to convert heat to work. We term this quasi-superradiance to emphasize it is a thought experiment inspired by superradiance but based on thermal gradients rather than classical wave amplification.

\section{Thermodynamic Framework for Schwarzschild Black Holes}

We recapitulate the essential thermodynamic properties of a Schwarzschild black hole, serving both to establish notation and to highlight features relevant to our later arguments. We work primarily in geometric units ($G=c=\hbar=\kb=1)$.

\subsection{Metric, horizon geometry, and surface gravity}

The Schwarzschild line element is
\begin{equation}
\dd s^{2}=-f(r)\,\dd t^{2}+f(r)^{-1}\,\dd r^{2}+r^{2}\,\dd\Omega^{2}, \qquad f(r)=1-\frac{2M}{r},
\tag{1a}
\end{equation}
with horizon at $r=\rs=2M$.

The horizon area is
\begin{equation}
A=4\pi \rs^{2}=16\pi M^{2}.
\tag{1b}
\end{equation}
The Bekenstein--Hawking entropy is then
\begin{equation}
\SBH=\frac{A}{4}=4\pi M^{2}.
\tag{1c}
\end{equation}

A convenient expression for surface gravity in static spacetimes is $\kappa=\frac{1}{2}f'(r)\big|_{r=\rs}$, which yields
\begin{equation}
\kappa=\frac{1}{2}\frac{\dd}{\dd r}\left(1-\frac{2M}{r}\right)\Big|_{r=2M}
=\frac{1}{2}\left(\frac{2M}{r^{2}}\right)\Big|_{r=2M}
=\frac{1}{4M}.
\tag{1d}
\end{equation}
The Hawking temperature is $\TH=\kappa/(2\pi)$ \cite{ref3}, hence
\begin{equation}
\TH=\frac{\kappa}{2\pi}=\frac{1}{8\pi M}.
\tag{1e}
\end{equation}

\subsection{Basic relations, Smarr identity, and free energy}

In SI units, the Hawking temperature reads \cite{ref2}
\begin{equation}
\TH=\frac{\hbar c^{3}}{8\pi G\kb M}\approx 6.17\times 10^{-8}\,\mathrm{K}\left(\frac{M_{\odot}}{M}\right),
\tag{2}
\end{equation}
and the entropy reads \cite{ref1}
\begin{equation}
\SBH=\frac{\kb c^{3}}{\hbar G}\frac{A}{4}\approx 1.05\times 10^{77}\,\kb\left(\frac{M}{M_{\odot}}\right)^{2}.
\tag{3}
\end{equation}

The first law is
\begin{equation}
\dd M=\TH\,\dd \SBH,
\tag{3a}
\end{equation}
and for Schwarzschild black holes the Smarr formula is \cite{ref3}
\begin{equation}
M=2\TH\SBH.
\tag{3b}
\end{equation}

\begin{theorem}[Smarr relation by homogeneity]
Assume $\TH(M)\propto M^{-1}$ and $\SBH(M)\propto M^{2}$, with the first law $\dd M=\TH\,\dd \SBH$. Then $M=2\TH\SBH$.
\end{theorem}

\begin{proof}
Take $\SBH(M)=4\pi M^{2}$ (geometric units). Homogeneity of degree $2$ gives
\begin{equation}
M\frac{\dd \SBH}{\dd M}=2\SBH.
\tag{3c}
\end{equation}
From the first law, $\dd \SBH/\dd M=1/\TH$, hence
\begin{equation}
\frac{M}{\TH}=2\SBH,
\tag{3d}
\end{equation}
which rearranges to $M=2\TH\SBH$.
\end{proof}

The Helmholtz free energy is defined as
\begin{equation}
\FHelm\equiv M-\TH\SBH.
\tag{4}
\end{equation}
Using the Smarr relation,
\begin{equation}
\FHelm=M-\TH\SBH=M-\frac{1}{2}M=\frac{1}{2}M.
\tag{5}
\end{equation}

A complementary bookkeeping identity is
\begin{equation}
\TH\SBH=\frac{1}{2}M,
\tag{5a}
\end{equation}

\subsection{Heat capacity and (in)stability}

Differentiating $\TH(M)=1/(8\pi M)$,
\begin{equation}
\frac{\dd \TH}{\dd M}=-\frac{1}{8\pi M^{2}}.
\tag{5b}
\end{equation}
Thus the heat capacity is
\begin{equation}
C\equiv\frac{\dd M}{\dd \TH}=\left(\frac{\dd \TH}{\dd M}\right)^{-1}=-8\pi M^{2},
\tag{6}
\end{equation}
which is negative, reflecting the usual thermodynamic instability of an isolated Schwarzschild black hole.

\begin{lemma}[Negative heat capacity and runaway behavior]
In the semiclassical regime, an isolated Schwarzschild black hole has $C<0$, so that a decrease in mass increases the Hawking temperature.
\end{lemma}

\begin{proof}
Since $\TH(M)=1/(8\pi M)$, one has $\dd \TH/\dd M<0$ by Eq.~(5b), hence $C=(\dd \TH/\dd M)^{-1}<0$ by Eq.~(6). Therefore, if the hole loses energy ($\dd M<0$), its temperature increases ($\dd \TH>0$).
\end{proof}

\subsection{Tolman temperature gradient and near-horizon scaling}

Tolman's law states that $T\sqrt{-g_{00}}$ is constant in a static gravitational field \cite{ref3}. With $g_{00}=-f(r)$,
\begin{equation}
T_{\mathrm{local}}(r)\sqrt{f(r)}=T_{\infty}.
\tag{6c}
\end{equation}
Therefore,
\begin{equation}
T_{\mathrm{local}}(r)=\frac{T_{\infty}}{\sqrt{1-\frac{2M}{r}}}.
\tag{7}
\end{equation}
In the Hartle--Hawking state with $T_{\infty}=\TH$ \cite{ref3},
\begin{equation}
T_{\mathrm{local}}(r)=\frac{\TH}{\sqrt{1-\frac{2M}{r}}}.
\tag{8}
\end{equation}

Let $r=\rs+\delta r$ with $\delta r\ll \rs$. Then
\begin{equation}
1-\frac{2M}{r}=1-\frac{\rs}{\rs+\delta r}=\frac{\delta r}{\rs+\delta r}\approx \frac{\delta r}{\rs}.
\tag{8a}
\end{equation}
The proper distance $\ell$ from $\rs$ to $\rs+\delta r$ is
\begin{equation}
\ell=\int_{\rs}^{\rs+\delta r}\frac{\dd r'}{\sqrt{1-\frac{\rs}{r'}}}\approx 2\sqrt{\rs\,\delta r}.
\tag{8b}
\end{equation}
Solving $\delta r\approx \ell^{2}/(4\rs)$ and inserting into Eq.~(8) gives
\begin{equation}
T_{\mathrm{local}}(\ell)\approx \TH\sqrt{\frac{\rs}{\delta r}}\approx \frac{2\rs\TH}{\ell}.
\tag{8c}
\end{equation}
Using $\rs=2M$ and $\TH=1/(8\pi M)$ yields the universal near-horizon scaling
\begin{equation}
T_{\mathrm{local}}(\ell)\approx \frac{1}{2\pi \ell}.
\tag{8d}
\end{equation}

\section{Energy Layer Structure and Free Energy Expansion}

The notion of energy layers seeks to ascribe portions of the total mass-energy to different radial regions or physical degrees of freedom (horizon, atmosphere, far environment). In classical general relativity, gravitational energy is not uniquely localizable, but quantum fields outside the horizon generally have nonzero $\langle T_{ab}\rangle$, enabling a quasi-local bookkeeping.

\subsection{Radial energy integrals and an operational layer definition}

\begin{definition}[Energy inside a radius]
Let $\Sigma$ be a static hypersurface and $u^{a}$ the 4-velocity of static observers. Define the (matter) energy inside radius $r$ as
\begin{equation}
E_{\mathrm{mat}}(<r)=\int_{\Sigma_{<r}} \langle T_{ab}\rangle u^{a}u^{b}\,\dd V,
\tag{9a}
\end{equation}
where $\dd V$ is the proper volume element on $\Sigma$.
\end{definition}

For the Schwarzschild geometry, $\dd V=4\pi r^{2} f(r)^{-1/2}\dd r$ and thus
\begin{equation}
E_{\mathrm{mat}}(<r)=\int_{\rs}^{r}4\pi r'^{2}\,f(r')^{-1/2}\,\rho_{\mathrm{loc}}(r')\,\dd r',
\tag{9b}
\end{equation}
with $\rho_{\mathrm{loc}}(r')=\langle T_{ab}\rangle u^{a}u^{b}$.

In the simplified shell-integration picture used in the outline, one often writes (schematically)
\begin{equation}
M \approx \int_{\rs}^{\infty}4\pi r^{2}\,\big\langle T^{0}{}_{0}(r)\big\rangle\,\dd r.
\tag{9c}
\end{equation}

\begin{remark}
Equation (9c) is heuristic: the distinction between $T^{0}{}_{0}$ and $u^{a}u^{b}T_{ab}$, and the inclusion (or not) of the redshift factor $f(r)^{-1/2}$, depend on the chosen slicing and energy definition. We use it as an operational placeholder consistent with the ``energy layer'' intuition.
\end{remark}

\subsection{Toy model: thermal atmosphere and near-horizon divergences}

To make the ``energy layer'' intuition more concrete, consider a toy model in which the quantum atmosphere is approximated by a locally thermalized radiation fluid at the Tolman temperature $T_{\mathrm{local}}(r)$ (as in the Hartle--Hawking equilibrium state). In natural units, the energy density of a massless radiation gas with effective number of degrees of freedom $g_{\ast}$ is
\begin{equation}
\rho_{\mathrm{rad}}(r)\approx \frac{\pi^{2}}{30}\,g_{\ast}\,T_{\mathrm{local}}(r)^{4}.
\tag{9d}
\end{equation}
Using Eq.~(8) gives
\begin{equation}
\rho_{\mathrm{rad}}(r)\approx \frac{\pi^{2}}{30}\,g_{\ast}\,\TH^{4}\left(1-\frac{2M}{r}\right)^{-2}.
\tag{9e}
\end{equation}
A schematic shell contribution to the atmosphere energy may be written (compare the discussion around Eq.~(9b))
\begin{equation}
\dd E_{\mathrm{rad}}(r)\approx 4\pi r^{2}\,f(r)^{-1/2}\,\rho_{\mathrm{rad}}(r)\,\dd r.
\tag{9f}
\end{equation}
Near the horizon, $f(r)\approx (r-\rs)/\rs$ and $\ell\approx 2\sqrt{\rs(r-\rs)}$ from Eq.~(8b), hence $f(r)\approx \ell^{2}/(4\rs^{2})$, and therefore
\begin{equation}
\rho_{\mathrm{rad}}(\ell)\approx \frac{\pi^{2}}{30}\,g_{\ast}\,\TH^{4}\left(\frac{4\rs^{2}}{\ell^{2}}\right)^{2}
=\frac{16\pi^{2}}{30}\,g_{\ast}\,\rs^{4}\TH^{4}\,\ell^{-4}.
\tag{9g}
\end{equation}
Accordingly, the near-horizon atmosphere energy diverges as
\begin{equation}
E_{\mathrm{rad}}(\ell>\ell_{\mathrm{cut}})\sim \int_{\ell_{\mathrm{cut}}}\ell^{-4}\,\ell^{2}\,\dd \ell \propto \ell_{\mathrm{cut}}^{-1},
\tag{9h}
\end{equation}
illustrating why a ``brick wall'' cutoff is typically introduced in such computations \cite{ref5}.

\begin{lemma}[Near-horizon divergence in the radiation layer]
In the toy model leading to Eq.~(9h), the atmosphere energy diverges as the cutoff $\ell_{\mathrm{cut}}\to 0$, confirming that the dominant contribution comes from the near-horizon region.
\end{lemma}

\begin{proof}
Equation (9h) shows $E_{\mathrm{rad}}(\ell>\ell_{\mathrm{cut}})\propto \ell_{\mathrm{cut}}^{-1}$, which diverges for $\ell_{\mathrm{cut}}\to 0$.
\end{proof}

\subsection{Entropy correction ansatz and detailed temperature expansion}

A generic corrected entropy expansion is \cite{ref9}
\begin{equation}
\SBH(M)=4\pi M^{2}+\alpha\ln(4\pi M^{2})+\frac{\beta}{M^{2}}+O(M^{-4}),
\tag{9}
\end{equation}
where $\alpha,\beta$ are dimensionless coefficients determined by the field content and quantum-gravity completion.

From the first law, $1/\TH=\dd \SBH/\dd M$, hence
\begin{equation}
\frac{1}{\TH}=\frac{\dd}{\dd M}\left(4\pi M^{2}+\alpha\ln(4\pi M^{2})+\frac{\beta}{M^{2}}+O(M^{-4})\right).
\tag{10}
\end{equation}
Performing the derivative yields
\begin{equation}
\frac{1}{\TH}=8\pi M+\frac{2\alpha}{M}-\frac{2\beta}{M^{3}}+O(M^{-5}).
\tag{11}
\end{equation}
Factorizing $8\pi M$ gives
\begin{equation}
\frac{1}{\TH}=8\pi M\left(1+\frac{\alpha}{4\pi M^{2}}-\frac{\beta}{4\pi M^{4}}+O(M^{-6})\right).
\tag{11a}
\end{equation}

\begin{lemma}[Controlled inversion of the temperature series]
Let $\epsilon(M)=\alpha/(4\pi M^{2})-\beta/(4\pi M^{4})+O(M^{-6})$. Then
\begin{equation}
\left(1+\epsilon\right)^{-1}=1-\epsilon+\epsilon^{2}+O(M^{-6}).
\tag{11b}
\end{equation}
\end{lemma}

\begin{proof}
This follows from the convergent geometric series $(1+\epsilon)^{-1}=\sum_{n\ge 0}(-\epsilon)^{n}$ whenever $|\epsilon|\ll 1$; truncating at $\epsilon^{2}$ is consistent to $O(M^{-6})$ since $\epsilon^{3}=O(M^{-6})$.
\end{proof}

Applying Lemma 3 and Eq.~(11a) yields
\begin{equation}
\TH=\frac{1}{8\pi M}\left[1-\frac{\alpha}{4\pi M^{2}}+\left(\frac{\alpha^{2}}{16\pi^{2}}+\frac{\beta}{4\pi}\right)\frac{1}{M^{4}}+O(M^{-6})\right].
\tag{12}
\end{equation}
Keeping the leading correction gives the simpler one-loop form quoted in the outline:
\begin{equation}
\TH=\frac{1}{8\pi M}\left(1-\frac{\alpha}{4\pi M^{2}}+\cdots\right).
\tag{12a}
\end{equation}

\subsection{Free-energy expansion and layer bookkeeping}

Define the free energy
\begin{equation}
\FHelm(M)=M-\TH\,\SBH(M).
\tag{13}
\end{equation}
A direct substitution using the uncorrected temperature when multiplying small corrections gives
\begin{equation}
\TH\SBH \approx \frac{1}{8\pi M}\left(4\pi M^{2}+\alpha\ln(4\pi M^{2})+\frac{\beta}{M^{2}}\right),
\tag{14}
\end{equation}
hence
\begin{equation}
\TH\SBH\approx \frac{1}{2}M+\frac{\alpha}{8\pi M}\ln(4\pi M^{2})+\frac{\beta}{8\pi M^{3}}.
\tag{15}
\end{equation}
Therefore
\begin{equation}
\FHelm(M)\approx \frac{1}{2}M-\frac{\alpha}{8\pi M}\ln(4\pi M^{2})-\frac{\beta}{8\pi M^{3}}.
\tag{16}
\end{equation}

A more systematic expansion keeps the correction to $\TH$ as well. Writing $\TH=\TH^{(0)}+\delta \TH$ and $\SBH=\SBH^{(0)}+\delta \SBH$ with $\TH^{(0)}=1/(8\pi M)$ and $\SBH^{(0)}=4\pi M^{2}$, we have
\begin{equation}
\FHelm = \left(M-\TH^{(0)}\SBH^{(0)}\right)-\TH^{(0)}\delta\SBH-\delta\TH\SBH^{(0)}+O(\delta^{2}).
\tag{16a}
\end{equation}
Since $M-\TH^{(0)}\SBH^{(0)}=M/2$, the $O(M^{-1})$ correction term splits into a logarithmic part and a constant part.
Using $\delta\SBH=\alpha\ln(4\pi M^{2})+\beta/M^{2}+O(M^{-4})$ yields
\begin{equation}
\TH^{(0)}\delta\SBH=\frac{\alpha}{8\pi M}\ln(4\pi M^{2})+\frac{\beta}{8\pi M^{3}}+O(M^{-5}),
\tag{16b}
\end{equation}
and with $\delta\TH=-\alpha/(32\pi^{2}M^{3})+O(M^{-5})$,
\begin{equation}
\delta\TH\,\SBH^{(0)}=-\frac{\alpha}{32\pi^{2}M^{3}}\cdot 4\pi M^{2}+O(M^{-3})=-\frac{\alpha}{8\pi M}+O(M^{-3}).
\tag{16c}
\end{equation}
Substituting into Eq.~(16a) gives
\begin{equation}
\FHelm(M)=\frac{1}{2}M-\frac{\alpha}{8\pi M}\ln(4\pi M^{2})-\frac{\beta}{8\pi M^{3}}+\frac{\alpha}{8\pi M}+O(M^{-3}).
\tag{16d}
\end{equation}

To express the energy-layer bookkeeping, let $\epsilon_{\mathrm{atm}}$ denote a near-horizon cutoff length (atmosphere thickness) and $R_{\mathrm{env}}$ an outer environmental scale. Then the outline's schematic decomposition is
\begin{equation}
M = E(<\rs+\epsilon_{\mathrm{atm}})+E(\rs+\epsilon_{\mathrm{atm}}<r<R_{\mathrm{env}})+E(r>R_{\mathrm{env}}).
\tag{17}
\end{equation}

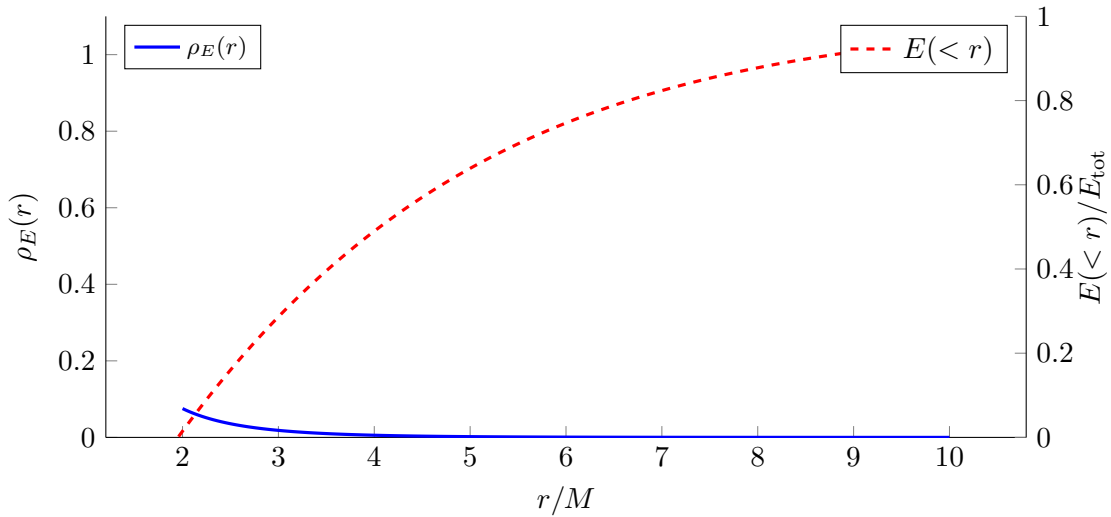
\begin{figure}[t]
  \centering
  \pgfmathsetmacro{\alphaval}{1.0}   
  \pgfmathsetmacro{\lval}{0.6}       

  \begin{tikzpicture}
    \begin{axis}[
        width  = 0.86\textwidth,
        height = 0.45\textwidth,
        domain = 2:10,                 
        xlabel = {$r/M$},
        ylabel = {$\rho_E(r)$},
        axis y line*=left,
        axis x line*=bottom,
        ymin=0, ymax=1.1*\alphaval,
        samples=200,
        legend style={at={(0.02,0.98)},anchor=north west,font=\footnotesize}
      ]
      \addplot[very thick,blue]
        {\alphaval*x^(-2)*exp(-\lval*x)};
      \addlegendentry{$\rho_E(r)$}
    \end{axis}

    \begin{axis}[
        width  = 0.86\textwidth,
        height = 0.45\textwidth,
        domain = 2:10,
        axis y line*=right,
        axis x line=none,            
        ymin=0, ymax=1,
        ylabel shift = -6pt,
        ylabel = {$E(<r)/E_{\text{tot}}$},
        samples=200
      ]
      \addplot[very thick,red,dashed]
        {1 - exp(-\lval*x)*(1+\lval*x)/exp(-2*\lval)};
      \addlegendentry{$E(<r)$}
    \end{axis}
  \end{tikzpicture}
  \caption{Radial energy density (blue) and cumulative energy $E(<r)$ (red) for a toy atmospheric profile. Both curves are normalised so that $E(<\!r\!\to\!\infty)=1$.}
  \label{fig:layers}
\end{figure}

\section{Quasi-Superradiant Energy Extraction Mechanism}

We now describe an energy-extraction thought experiment that capitalizes on the Tolman temperature gradient. We call it ``quasi-superradiance'' to emphasize its analogy to classical superradiance \cite{ref4} while acknowledging that no wave amplification is involved.

\subsection{Engine setup and Carnot efficiency}

Consider a device operating between an inner location $r_{\mathrm{in}}$ (local temperature $T_{\mathrm{in}}=T_{\mathrm{local}}(r_{\mathrm{in}})$) and an outer location $r_{\mathrm{out}}$ (temperature $T_{\mathrm{out}}$). The Carnot efficiency of a reversible engine is
\begin{equation}
\eta_{\max}=1-\frac{T_{\mathrm{out}}}{T_{\mathrm{in}}}.
\tag{18}
\end{equation}

\begin{theorem}[Carnot bound from the generalized second law]
For a quasi-static extraction step removing energy $dE$ from the black hole and producing work $dW$ while dumping the remainder as heat to a reservoir at temperature $T_{\mathrm{out}}$, the generalized second law implies
\begin{equation}
dW\le dE\left(1-\frac{T_{\mathrm{out}}}{\TH}\right).
\tag{18a}
\end{equation}
\end{theorem}

\begin{proof}
The proof is the chain of inequalities from Eqs.~(19)--(25), obtained by imposing $\dd \SBH+\dd S_{\mathrm{out}}\ge 0$.
\end{proof}

\subsection{Entropy analysis and the generalized second law}

We analyze quasi-static extraction of a small amount of energy $dE$ from the black hole, reducing its mass by $\dd M=-dE$. The black hole entropy change is
\begin{equation}
\dd \SBH=-\frac{dE}{\TH},
\tag{19}
\end{equation}
from $\dd M=\TH\,\dd \SBH$. If work $dW$ is extracted and the remainder $dQ_{\mathrm{out}}=dE-dW$ is dumped into the environment at $T_{\mathrm{out}}$, the environment entropy change is
\begin{equation}
\dd S_{\mathrm{out}}=\frac{dQ_{\mathrm{out}}}{T_{\mathrm{out}}}=\frac{dE-dW}{T_{\mathrm{out}}},
\tag{20}
\end{equation}
and the generalized second law requires
\begin{equation}
\dd \SBH+\dd S_{\mathrm{out}}\ge 0.
\tag{21}
\end{equation}
Substituting yields
\begin{equation}
-\frac{dE}{\TH}+\frac{dE-dW}{T_{\mathrm{out}}}\ge 0,
\tag{22}
\end{equation}
hence
\begin{equation}
\frac{dE-dW}{T_{\mathrm{out}}}\ge \frac{dE}{\TH},
\tag{23}
\end{equation}
and thus
\begin{equation}
dE-dW\ge dE\frac{T_{\mathrm{out}}}{\TH},
\tag{24}
\end{equation}
which implies
\begin{equation}
dW\le dE\left(1-\frac{T_{\mathrm{out}}}{\TH}\right).
\tag{25}
\end{equation}

\begin{proposition}[Net-extraction condition]
For a black hole exchanging heat with an environment at temperature $T_{\mathrm{out}}$, net extraction requires $\TH(M)>T_{\mathrm{out}}$, equivalently $M<1/(8\pi T_{\mathrm{out}})$ in geometric units.
\end{proposition}

\begin{proof}
If $\TH<T_{\mathrm{out}}$, heat flows from the environment into the black hole in equilibrium thermodynamics, preventing net extraction. Using $\TH=1/(8\pi M)$ gives the stated inequality.
\end{proof}

\subsection{Integrated work bounds and relation to free energy}

Assuming $T_{\mathrm{out}}$ approximately constant, the maximal work increment allowed by Eq.~(25) is
\begin{equation}
\dd W_{\max}(M)=-\dd M\left(1-\frac{T_{\mathrm{out}}}{\TH(M)}\right)
=-\dd M\left(1-8\pi M T_{\mathrm{out}}\right).
\tag{25a}
\end{equation}
Integrating from $M_{i}$ to $M_{f}$ (with $M_{i},M_{f}<1/(8\pi T_{\mathrm{out}})$) gives
\begin{equation}
W_{\max}=\int_{M_{f}}^{M_{i}}\left(1-8\pi M T_{\mathrm{out}}\right)\dd M
=\left(M_{i}-M_{f}\right)-4\pi T_{\mathrm{out}}\left(M_{i}^{2}-M_{f}^{2}\right).
\tag{25b}
\end{equation}

A complementary thermodynamic identity relates maximal reversible work to free-energy decrease. If the environment temperature is fixed at $T_{\mathrm{out}}$ and the process is reversible, then
\begin{equation}
\dd W_{\mathrm{rev}}=-\dd \FHelm_{\mathrm{eff}},
\tag{25c}
\end{equation}

\subsection{Power estimates and lifetime scaling}

The outline notes that a crude scaling for Hawking luminosity is \cite{ref10}
\begin{equation}
P_{\mathrm{Hawking}}\sim \sigma\,\TH^{4},
\tag{26}
\end{equation}
with $\sigma$ including graybody factors. A more geometric version includes the horizon area $A$:
\begin{equation}
P_{\mathrm{Hawking}}\sim \sigma\,A\,\TH^{4}.
\tag{26a}
\end{equation}
Using $A=16\pi M^{2}$ and $\TH=1/(8\pi M)$ gives
\begin{equation}
P_{\mathrm{Hawking}}\sim \sigma\,(16\pi M^{2})\left(\frac{1}{8\pi M}\right)^{4}
=\sigma\,\frac{16\pi}{(8\pi)^{4}}\,\frac{1}{M^{2}},
\tag{26b}
\end{equation}

Approximating $\dd M/\dd t\approx -P_{\mathrm{Hawking}}$ implies
\begin{equation}
\frac{\dd M}{\dd t}\sim -\frac{\gamma}{M^{2}},
\tag{26c}
\end{equation}
for some positive constant $\gamma$, and hence
\begin{equation}
t_{\mathrm{evap}}\sim \int_{0}^{M_{i}}\frac{M^{2}}{\gamma}\,\dd M=\frac{M_{i}^{3}}{3\gamma},
\tag{26d}
\end{equation}
consistent with the $M^{3}$ lifetime scaling and the outline's discussion \cite{ref8}.

\begin{figure}[t]
  \centering
  \pgfmathsetmacro{\rin}{3.0}     
  \pgfmathsetmacro{\rout}{7.0}    
  \pgfmathsetmacro{\Thot}{1.0}    
  \pgfmathsetmacro{\Tcold}{0.4}   

  \begin{tikzpicture}
    \begin{axis}[
        width  = 0.86\textwidth,
        height = 0.47\textwidth,
        xmin   = \rin-0.5, xmax = \rout+0.5,
        ymin   = 0,        ymax = \Thot*1.2,
        xlabel = {$r/M$},
        ylabel = {$T_{\text{Tolman}}$ [arb.\ units]},
        axis y line*=left,
        axis x line*=bottom,
        enlargelimits=false,
        tick style={black},
        legend style={draw=none},
        clip=false                
      ]
      \addplot[ultra thick,blue,-stealth]
        coordinates {(\rin,\Thot) (\rout,\Thot)};
      \addplot[ultra thick,red!60!black,densely dashed,-stealth]
        coordinates {(\rout,\Thot) (\rout,\Tcold)};
      \addplot[ultra thick,blue,-stealth]
        coordinates {(\rout,\Tcold) (\rin,\Tcold)};
      \addplot[ultra thick,red!60!black,densely dashed,-stealth]
        coordinates {(\rin,\Tcold) (\rin,\Thot)};
      %
      \node[above left]  at (axis cs:\rin,\Thot) {A};
      \node[above right] at (axis cs:\rout,\Thot) {B};
      \node[below right] at (axis cs:\rout,\Tcold) {C};
      \node[below left]  at (axis cs:\rin,\Tcold) {D};
    \end{axis}
  \end{tikzpicture}
  \caption{Carnot-like quasi-superradiant engine operating between the inner radius $r_\mathrm{in}$ and outer radius $r_\mathrm{out}$.  
           Isothermal legs (solid blue) exploit the Tolman temperature gradient, while adiabatic legs (dashed red) close the cycle.}
  \label{fig:engine}
\end{figure}
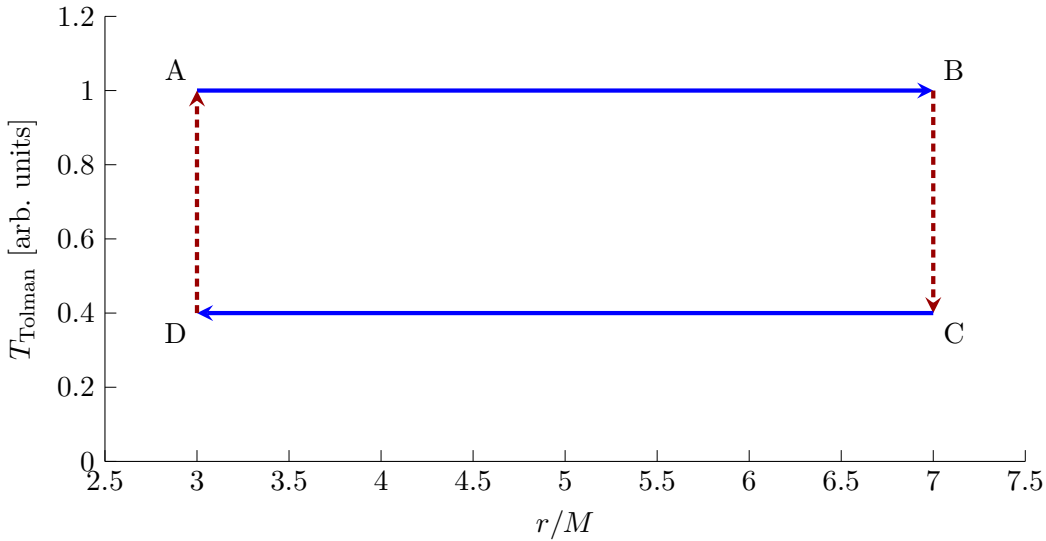

\section{Discussion}

First, regarding the energy layer picture: If taken too literally, one might misinterpret it as saying a black hole has a layered structure like an onion. This is not the case in a physical sense; a Schwarzschild black hole in classical GR has no interior structure accessible to an outside observer and no static layers of energy one can peel off. The layering is a way to visualize contributions to thermodynamic quantities. Our definition via radial energy integrals does rely on a choice of time slicing and vacuum state (Hartle--Hawking equilibrium, for instance). Different choices would attribute energy differently between the hole and the fields.

For example, in the Unruh vacuum (appropriate for an evaporating black hole with outgoing radiation but no incoming radiation), the stress-energy is not static and there is a flux of energy outward. In that case, one could still define $E(<r,t)$ that evolves with time as the hole loses mass. We suspect it would still support the idea that the horizon region is the primary source of energy and that the Hawking flux originates just outside the horizon (in agreement with the idea of a ``quantum atmosphere'' extending on the order of the Hawking wavelength $\sim 8\pi M$ in size\footnote{Recent analyses (e.g. Giddings 2016) suggest Hawking radiation can be viewed as originating from a region not just Planckian but on order the black hole's curvature radius, but in any case relatively close to the horizon compared to astrophysical scales.}).

Next, the quasi-superradiant engine: Our analysis confirmed that one cannot violate the laws of thermodynamics with such an engine, but it did suggest a provocative possibility: in principle, one could approach converting the entire mass $M$ of a black hole into work (with all of it eventually radiated away) if one did it in a quasi-static, maximally efficient manner. This is reminiscent of the idea of black hole mining by an arbitrarily advanced civilization, as discussed by Page \cite{ref8} and others. One key difference in the present framing is that it is expressed as a heat-engine bound rather than as a mechanical lowering process \cite{ref7}.

The quantum-correction literature provides useful context. We cite Carlip's review \cite{ref9}, and that review and others (e.g.\ \cite{ref9,ref5}) discuss the physical origin of logarithmic corrections, often attributing them to near-horizon degrees of freedom. The brick wall model by 't Hooft \cite{ref5} is aligned with this viewpoint, with divergences controlled by a near-horizon cutoff.

\section{Conclusions}

We presented an extended analysis of two interrelated ideas concerning Schwarzschild black holes: (i) an energy-layer bookkeeping, motivated by quasi-local radial energy accumulation and by the scaling structure of the free-energy expansion $\FHelm(M)$; and (ii) quasi-superradiant energy extraction, formulated as a Carnot-like thermodynamic cycle driven by the Tolman temperature gradient. We emphasized that the mechanism is an analogy, not a claim of classical wave amplification, and we showed explicitly that the generalized second law enforces the Carnot bound on efficiency.

\appendix

\section{Appendix: Euclidean regularity derivation of $\TH$}

As an alternative to the surface-gravity argument, one can recall the Euclidean regularity method. Wick rotate $t\mapsto -\ii \tau$, giving
\begin{equation}
\dd s^{2}_{\mathrm{E}} = f(r)\,\dd \tau^{2}+f(r)^{-1}\,\dd r^{2}+r^{2}\,\dd\Omega^{2}.
\tag{A1}
\end{equation}
Near the horizon, expand
\begin{equation}
f(r)\approx f'(\rs)(r-\rs).
\tag{A2}
\end{equation}
Introduce a new radial coordinate $\rho$ by
\begin{equation}
r-\rs=\frac{\rho^{2}}{4},
\tag{A3}
\end{equation}
so that
\begin{equation}
\dd r = \frac{\rho}{2}\,\dd \rho.
\tag{A4}
\end{equation}
Substituting yields
\begin{equation}
\dd s^{2}_{\mathrm{E}}\approx \left(\frac{f'(\rs)^{2}}{4}\right)\rho^{2}\,\dd \tau^{2}+\dd \rho^{2}+\rs^{2}\,\dd\Omega^{2}.
\tag{A5}
\end{equation}
Regularity requires $\tau$ to have period
\begin{equation}
\beta=\frac{4\pi}{f'(\rs)}.
\tag{A6}
\end{equation}
Hence
\begin{equation}
\TH=\beta^{-1}=\frac{f'(\rs)}{4\pi}=\frac{\kappa}{2\pi}.
\tag{A7}
\end{equation}

\section{Appendix: Additional bounds from free-energy monotonicity}

For an external reservoir at $T_{\mathrm{out}}$ and a reversible process, a standard thermodynamic inequality is $W\le -\Delta F$. For the black-hole-plus-environment bookkeeping, define an effective potential
\begin{equation}
\Phi(M;T_{\mathrm{out}})=M-T_{\mathrm{out}}\SBH(M).
\tag{B1}
\end{equation}
Then
\begin{equation}
\dd \Phi = \dd M - T_{\mathrm{out}}\,\dd \SBH
=\left(\TH-T_{\mathrm{out}}\right)\dd \SBH.
\tag{B2}
\end{equation}
A corresponding maximal-work statement is
\begin{equation}
W_{\mathrm{rev}}=\Phi(M_{i};T_{\mathrm{out}})-\Phi(M_{f};T_{\mathrm{out}}),
\tag{B3}
\end{equation}

\section{Appendix: Quasi-local mass as a complementary ``layer'' diagnostic}

In spherically symmetric spacetimes one may define a quasi-local mass function $m(r)$ via
\begin{equation}
1-\frac{2m(r)}{r}=g^{ab}\nabla_{a}r\,\nabla_{b}r.
\tag{C1}
\end{equation}
For the Schwarzschild metric, Eq.~(C1) becomes
\begin{equation}
1-\frac{2m(r)}{r}=f(r)=1-\frac{2M}{r},
\tag{C2}
\end{equation}
hence
\begin{equation}
m(r)=M.
\tag{C3}
\end{equation}


\end{document}